\documentclass[draft]{svjour3}
\usepackage[T1]{fontenc}
\usepackage{textcomp}
\usepackage{mathptmx}

\title{Normality and Finite-State Dimension of Liouville Numbers}
\author{Satyadev Nandakumar  \and 
       Santhosh Kumar Vangapelli}
\institute{Department of Computer Science,
Indian Institute of Technology Kanpur,
Kanpur, Uttar Pradesh, India.
\and
Google Inc.,
Hyderabad, Andhra Pradesh, India.}
\titlerunning{Normality of Liouville Numbers}
\authorrunning{S. Nandakumar and Santosh Kumar Vangepalli}
\date\today

\usepackage{amsfonts}
\usepackage{amsmath} 
\usepackage{amssymb}
\usepackage{enumerate}

\newcommand{\N}{\mathbb{N}}
\newcommand{\Z}{\mathbb{Z}}

\renewcommand{\dim}{{\mathrm{dim}}}

\begin{document}
\maketitle
\begin{abstract}
Liouville numbers were the first class of real numbers which were
proven to be transcendental. It is easy to construct non-normal
Liouville numbers.  Kano \cite{Kano93} and Bugeaud \cite{Bugeaud02}
have proved, using analytic techniques, that there are normal
Liouville numbers. Here, for a given base $k \ge 2$, we give a new
construction of a Liouville number which is normal to the base
$k$. This construction is combinatorial, and is based on de Bruijn
sequences.

A real number in the unit interval is normal if and only if its
finite-state dimension is 1. We generalize our construction to prove
that for any rational $r$ in the closed unit interval, there is a
Liouville number with finite state dimension $r$. This refines
Staiger's result \cite{Stai02} that the set of Liouville numbers has
constructive Hausdorff dimension zero, showing a new quantitative
classification of Liouville numbers can be attained using finite-state
dimension. We also give a construction of a Liouville number normal in
finitely many bases, provided a generalized version of Artin's
conjecture holds.
\end{abstract}

\section{Introduction}
One of the important open questions in the study of normality is
whether any algebraic irrational number is normal. On the other hand,
it is known that there are normal transcendentals as well as
non-normal transcendentals. For example, Mahler has proved that the
Champernowne constant \cite{Mahler37} as well as the Thue-Morse
\cite{Mahler29} constant are transcendental. However, the Champernowne
constant is normal \cite{Cham33}, whereas the Thue-Morse constant is
not, and has finite-state dimension 0 \cite{Bachan05}.

Liouville numbers were the first class of numbers which were proven
transcendental. In this paper, we show that there are Liouville
numbers which are non-normal, and others which are normal. Indeed,
there are Liouville numbers of every rational finite-state dimension
between 0 and 1 (definitions follow in Section 3).

%% Careful!
Examples of non-normal Liouville numbers are well-known. Normal
Liouville numbers are harder to construct, and there are works by Kano
\cite{Kano93} and Bugeaud \cite{Bugeaud02} which establish the
existence of such numbers. Kano constructs, for any bases $a$ and $b$,
Liouville numbers which are normal in base $a$ but not in base
$ab$. Bugeaud gives a non-constructive proof using Fourier analytic
techniques that there are Liouville numbers which are absolutely
normal -- that is, normal in all bases. In this paper, we give a
combinatorial construction of a number that is normal to a given base
$b$. The construction is elementary.  Thus the Liouville numbers forms
a class of numbers whose transcendence is easy to establish, and which
contain simple examples of normal and non-normal numbers.

The set of normal numbers coincide exactly with the set of numbers
with finite-state dimension 1 \cite{SchnorrStimm72},
\cite{Dai:FSD}. We show that the combinatorial nature of our
construction lends itself to the construction of Liouville numbers of
any finite-state dimension. Thus we get a quantitative classification
of non-normal Liouville numbers. This classification is new, since the
set of Liouville numbers has classical Hausdorff dimension and even
effective Hausdorff Dimension zero \cite{Stai02}.

We begin with a survey briefly explaining Liouville's approximation
theorem and defining the class of Liouville numbers. Section
\ref{sec:liouville_disj} constructs a Liouville number which is
disjunctive -- that is to say, has all strings appearing in its base
$b$ expansion-- but is still not normal. The subsequent section gives
the construction of a normal Liouville number.

\section{Liouville's Constant and Liouville Numbers}
Liouville's approximation theorem says that algebraic irrationals are
inapproximable by rational numbers to arbitrary precisions. 
\begin{theorem} [Liouville's Theorem]
Let $\beta$ be a root of $f(x) = \sum_{j=0}^{n} a_j x^j \in \Z[x]$.
Then there is a constant $C_\beta$ such that for every pair of
integers $a$ and $b$, $b > 0$, we have
$\left| \beta - \frac{a}{b} \right| > \frac{C_\beta}{b^n}$.
\end{theorem}
%\index{Number!Algebraic}

Liouville then constructed the following provably irrational number
$ \psi = \sum_{i = 0}^{\infty} 10^{-i!}, $
and showed that it had arbitrarily good rational approximations in the
above sense, and therefore is a  transcendental number. In this paper,
since we consider base-2 expansions, we will show that 
$$\psi_1 = \sum_{i=1}^{\infty} 2^{-i!}$$
is a Liouville number.

\begin{definition}
A real number $\alpha$ in the unit interval is called a
\emph{Liouville number} if for all numbers $n$, there are numbers $p >
0$ and $q > 1$ such that $ \left| \alpha - \frac{p}{q} \right| <
\frac{1}{q^n}$. %\index{Number!Liouville}
\end{definition}

For every $n$, we have
$$\left| \psi_1 - \sum_{i = 1}^{n} \frac{1}{2^{i!}} \right| =
 \sum_{i = n+1}^{\infty} \frac{1}{2^{i!}}
= \sum_{i = (n+1)!}^{\infty} \frac{1}{2^i}
= \frac{1}{2^{(n+1)!-1}}
< \frac{1}{2^{n! n}}
= \frac{1}{q_n^n},
$$
where $q_n$ is the denominator of the finite sum  $\sum_{i = 1}^{n}
\frac{1}{2^{i!}}$. Thus $\psi_1$ is a Liouville number.

It is easy to see that the Liouville constant $\psi_1$ is not a normal
number - the sequence $111$ never appears in the decimal expansion of
$\psi_1$. It is natural to investigate whether all Liouville numbers
are non-normal.

This requires sharper observations than the one above.

\section{Disjunctive Liouville Sequences}
\label{sec:liouville_disj}
Hertling \cite{Hert96} has showed that there are disjunctive Liouville
numbers - that is, there are Liouville numbers whose base $r$
expansions have all possible $r$-alphabet strings. Staiger
strengthened this result to show that there are Liouville numbers
which are disjunctive in any base \cite{Stai02}. This shows that we
cannot rely on the above argument of absent strings to show
non-normality.

Here, to motivate the construction of a normal number in the next
section, we give a different construction of a different disjunctive
Liouville sequence. Consider $\psi_2 = \sum_{i=3}^{\infty}
\frac{i}{2^{i!}}$.  For any binary string $w$, we know that $1w$ (1
concatenated with $w$) is the binary representation of a number, hence
it appears in the binary expansion of $\psi_2$. Thus $\psi_2$ is a
disjunctive sequence.

It is easy to see that $\psi_2$ is not a normal number. At all large
enough prefix lengths of the form $n!$, there are at most $n(\lfloor
\log_2 n \rfloor + 1)$ ones - this follows from the fact that at most
$n$ unique non-zero numbers have appeared in the binary expansion of
$\psi_2$, and each of the numbers can be represented with at most
$\lfloor \log_2 n \rfloor + 1$ bits. Hence
$$\liminf_{n \to \infty} \frac{\left | \left\{ i : 0 \le i \le n-1
  \text{ and } \psi_2[i] = 1\right\}\right|}{n} \le \lim_{n \to \infty}
  \frac{n \left(\lfloor \log_2 n \rfloor + 1\right)}{n!} = 0,$$ 
which proves that $\psi_2$ is not normal.

However, $\psi_2$ is a Liouville number: For every $n$, there are
rationals with denominators of size $2^{n!}$ which satisfy the
Liouville criterion, as follows:
\begin{align*}
\left| \psi - \sum_{i = 3}^{n} \frac{i}{2^{i!}} \right| &=
 \sum_{i = n+1}^{\infty} \frac{i}{2^{i!}}
 <\, \sum_{k = 1}^{\infty} \frac{n+k}{2^{(n+1)! \cdot k}}
\end{align*}
Summing up the series, we obtain the inequality 
\begin{align*}
\sum_{k = 1}^{\infty} \frac{n+k}{2^{(n+1)! \cdot k}} &= 
 \frac{[n+1] \cdot 2^{(n+1)! - 1} + 1}{2^{((n+1)! - 1)\cdot2}}
 <\, \frac{[n+2]}{2^{(n+1)! - 1}} <\, \frac{1}{2^{(n!) \cdot n}}.
\end{align*}

\section{A Normal Liouville Number}
\label{sec:liouville_nrml}
Though the Liouville numbers constructed above were non-normal, there
are normal Liouville numbers. We give such a construction below, which
depends on \textsc{de Bruijn sequences} introduced by de Bruijn
\cite{deBruijn46} and Good \cite{Good46}, a standard tool in the study
of normality.\footnote{There are several historical forerunners of
this concept in places as varied as Sanskrit prosody and
poetics. However, the general construction for all bases and all
orders is not known to be ancient.}

\begin{definition}
Let $\Sigma$ be an alphabet with size $k$. A \emph{k-ary de Bruijn
sequence} $B(k,n)$ of order $n$, is a finite string for which every
possible string in $\Sigma^n$ appears exactly once.\footnote{For the
  string $x_0 x_1 \dots x_{k^n-1}$, we also consider subpatterns
  $x_{k^n-n+1} \dots x_{k^n-1} x_0$, $x_{k^n-n+2} \dots x_{k^n-1}
  x_0x_1$, and so on until $x_{k^n-1} x_0 \dots x_{n-2}$, obtained by
  ``wrapping around the string''.}
\end{definition}
%\index{Sequence!de Bruijn}

de Bruijn proved that such sequences exist for all $k$ and all orders
$n$. Since each de Bruijn sequence $B(k,n)$ contains each $n$-length
string exactly once, it follows that the length of $B(k,n)$ is exactly
$k^n$.

\subsection{Construction}
If $w$ is any string, we write $w^{i}$ for the string formed by
repeating $w$, $i$ times. In this section, we limit ourselves to the
binary alphabet $\Sigma = \{0,1\}$ even though the construction
generalizes to all alphabets.

Consider $\alpha \in [0,1)$ with binary expansion defined as follows.
$$\alpha = 0\;\;.\;\;B(2,1)^{1^1} B(2,2)^{2^2}B(2,3)^{3^3} \dots
B(2,i)^{i^i} \dots.$$

Informally, we can explain why this construction defines a normal
Liouville number, as follows. The Liouville numbers $\psi_1$ and
$\psi_2$ that we considered before, have prefixes that are mostly
zeroes. The density of 1s go asymptotically to zero as we consider
longer and longer prefixes. So it is fairly easy for a finite state
compressor to compress the data in a prefix. However, in this
construction, the repeating patterns employed are those which are
eventually hard for any given finite state compressor. This is why the
sequence could be normal.

Moreover, the transition in the patterns occur at prefix lengths of
the form $k^k$. By Stirling's approximation,
$$ k! \approxeq k^k e^{-k} \sqrt{2 \pi k}.$$ 
So the transitions in the pattern occur at prefix lengths similar to
that of $\psi_1$ and $\psi_2$, so it is reasonable to expect a
sequence of rationals approximating $\alpha$ that obeys the Liouville
criterion. We now make this argument more precise.

For any $i$, let $n_i = \sum_{m=1}^{i} m^m 2^m$. We call the part of
$\alpha[n_{i-1} \dots n_i - 1]$ as the $i^th$ stage of $\alpha$, which
consists of $i^i$ copies of $B(2,i)$. Thus $n_i$ denotes the length of
the prefix of $\alpha$ which has been defined at the end of the
$i^{\text{th}}$ stage. We have the following estimate for $n_i$.
\begin{align*}
n_i =\, \sum_{m=1}^{i} m^m 2^{m} 
    \quad<\quad i^i \sum_{m=1}^{i} 2^{m}
    =\, i^i \left[2^{i+1} - 2\right]  
    =\, i^i 2^{i+1} - 2i^{i}
    \quad<\quad 2 \left( i^i 2^i \right).
    \phantom{asdvad}
\end{align*}
Thus $n_i = O([2i]^i)$.

\begin{lemma}
$\alpha$ is a Liouville number. 
\end{lemma}
\begin{proof}
Consider the
rational number $\frac{p_i}{q_i}$ \footnote{not necessarily in the
  lowest form}  with a binary expansion which coincides with
$\alpha$ until the $i-1^{\text{st}}$ stage, followed by a recurring
block of $B(2,i)$. This rational number is 
$$\frac{\alpha[0 \dots n_{i-1}]}{2^{n_{i-1}}} + \frac{B(2,i)}{(2^{2^i}
  - 1)2^{n_{i-1}}},$$
obtained by evaluating the binary expansion as a geometric series.
The \emph{exponent} of the denominator of this rational number is 
$$2^i + n_{i-1} = 2^i+\,O\left([2(i-1)]^{i-1}\right) = O(n_{i-1}),$$ 
so the denominator of the rational is $2^{O\left(n_{i-1}\right)}$.

We add $i^i$ copies of $B(2,i)$ in the $i^{\text{th}}$ stage. Thus the
expansion of $\alpha$ and that of $r_i$ coincide for the first $n_i$
positions. So, $\alpha$ and $\frac{p_i}{q_i}$ are in the same dyadic
interval of length at most $2^{-n_i}$, and hence are within
$\frac{1}{2^{n_i}}$ of each other.

We have also that $n_i > i O\left([2(i-1)]^{(i-1)}\right)$, so
that 
$$\frac{1}{2^{n_i}} <
\frac{1}{\left(2^{O\left(n_{i-1}\right)}\right)^i}.$$  
Thus,
$$ \left| \alpha - \frac{p_i}{q_i}\right | < \frac{1}{q_i^i}.$$
Since this is true of any stage $i$, we can see that $\alpha$ is a
Liouville number.
\begin{flushright}\qed\end{flushright}
\end{proof}

\begin{lemma}
$\alpha$ is normal to the base 2.
\end{lemma}
\begin{proof}
Let us define $\textsc{count}: \Sigma^* \times \Sigma^* \to \N$ by
$$\textsc{count}(w; x) = \left| \left\{n \mid x[n \dots n+|w|-1] = w
\right\} \right|,$$ 
that is, the number of times $w$ occurs in $x$,
counting in a sliding block fashion.  For example, $00$ occurs twice
in $1000$. It is enough to show that for an
arbitrary binary string $w$ of length $m$, for all large enough
indices $j$,
$$\textsc{count}(w; \alpha[0\dots j-1]) = 2^{-m} j + o(j).$$

Let $j$ be a number greater than $n_{m}$. Every such index $j$ has a
number $i$ such that $n_i < j \le n_{i+1}$. 

We split the analysis into three phases, that of the prefix $\alpha[0
  \dots n_{m-1}-1]$, of the middle region $\alpha[n_{m-1} \dots n_i]$,
and of the suffix $\alpha[n_i \dots j-1]$.

The prefix $\alpha[0 \dots n_{m-1}]$ has a constant length that
depends on $w$ but not on $j$. Hence the discrepancy in the count of
$w$ due to this prefix, which is at most $n_{m-1}$, is
$o(j)$.

Since the number of times $w$ occurs in $B(2, m)$ is exactly $1$,
$$\textsc{count}(w; B(2,m)) = 2^{-m}\;|B(2,m)|.$$ 
Similarly, it is easy to see that for any $M > m$, by the properties
of the de Bruijn sequences,
$$\textsc{count}(w; B(2,M)) = 2^{-m}\;|B(2,M)|.$$
This observation is used in the following analysis of the middle part
and the suffix.

The part of of $\alpha$ in the stretch $n_{m-1} \dots n_m-1$
parses exactly into $m^m$ disjoint blocks of $B(2,m)$. Consequently,
$$\textsc{count}(w; \alpha[n_{m-1} \dots n_m]) = 2^{-m} (n_m -
n_{m-1}).$$

For all stages $k$ between $m-1$ and $i$,
$$\textsc{count}(w;\alpha[n_k \dots n_{k+1}-1]) = 2^{-m} (n_{k+1} -
n_k),$$
hence by a telescoping sum, 
$$\textsc{count}(w;\alpha[n_{m-1} \dots n_i-1]) = 2^{-m} (n_i -
n_{m-1}).$$

The suffix is formed during the $i+1^{\text{st}}$ stage of
construction of $\alpha$, and hence consists of $(i+1)^{(i+1)}$ copies
of $B(2,i+1)$. Let $j$ be such that 
$$ p. 2^{(i+1)} < j - n_i < (p+1) 2^{i+1}-1.$$
That is, $j$ falls within the $p+1^{\text{st}}$ copy of
$B(2,i+1)$. Then, 
$$j = n_i + p 2^{i+1} + o(j),$$ 
since the last term is at most $2^{i+1}-1$ and $n_i = \Omega(i^i)$.

Since $w$ is normally distributed in $\alpha[0 \dots n_i]$ and in each
of the $p$ preceding copies of $B(2,i+1)$, we have
$$\textsc{count}(w;\alpha[0 \dots j-1]) = 2^{-m} n_i + 2^{-m} p
2^{i+1} + o(j) = 2^{-m} j + o(j),$$ showing that $\alpha$ is normal.
\begin{flushright}\qed\end{flushright}
\end{proof}

\section{Finite State Dimension}
We now briefly give the block entropy characterization of finite state
dimension \cite{BHV05}. Finite-state dimension, or equivalently finite
state compressibility, is an asymptotic measure of information density
in a sequence measured by a finite-state automaton. This was
introduced by Dai, Lathrop, Lutz and Mayordomo \cite{Dai:FSD}. The
sequences with maximal density, are exactly the set of normal
sequences. These have finite-state dimension 1. A more detailed
study of the relationship between algorithmic randomness, normality
and Liouville numbers is found in Calude and Staiger
\cite{CaludeStaiger13}, and an investigation of the incompressibility
of Liouville numbers in a slightly different model is found in Calude,
Staiger and Stephan \cite{CaludeStaigerStephan14}. 

Let $\Omega$ be a nonempty finite set. Recall that the \emph{Shannon
  entropy} of a probability measure $\pi$ on $\Omega$ is
$$H(\pi) = \sum_{w \in \Omega} \pi(w) \log \frac{1}{\pi(w)},$$
where $0 \log \frac{1}{0} = 0$.

For nonempty strings $w, x \in \Sigma^+$, we write

$$\#(w,x) = \left| \left\{ m \le \frac{|x|}{|w|} - 1 \;\mid\;\; w =
x[m|w| \dots (m+1)|w| - 1\right\}] \right|.$$

That is, $\#(w,x)$ is the number of times a string $w$ of length $m$
occurs in $x$, when $x$ is parsed into disjoint blocks each of length
$m$.

For each infinite binary sequence, each positive integer $n$, and a
binary string $w$ of length $m$, the $n^{\text{th}}$ block frequency
of $w$ in $S$ is
$$\pi^{(m)}_{S,n}(w) = \frac{\pi(w,\;S[0 \dots n|w|-1])}{n}.$$

This defines a probability measure on $m$ long binary strings. The
normalized upper and lower block entropy rates of $S$ are
$$H^-_m(S) = \frac{1}{m} \liminf_{n \to \infty}
H\left(\pi^{(m)}_{S,n}\right)$$ 
and
$$H^+_m(S) = \frac{1}{m} \limsup_{n \to \infty}
H\left(\pi^{(m)}_{S,n}\right).$$

\begin{definition} 
Let $S \in \Sigma^\infty$. The \emph{finite state dimension} of $S$ is 
$$\dim_{\text{FS}}(S) = \inf_{m \in \Z^+} H^-_{m} (S),$$ and the
\emph{finite state strong dimension} of $S$ is
$$\dim_{\text{FS}}(S) = \inf_{m \in \Z^+} H^+_{m} (S).$$
\end{definition}

For purposes of the next section, we use a sliding block variant of
the block entropy. This is obtained by counting the frequency of
blocks in a sliding block fashion. Let $0<m<n$ be integers. The
frequency of an $m$-long block $w$ in an $n$-long string $x$ is
defined as
$$\frac{\left|\{i \mid w=x[i \dots i+m-1]\}\right|}{n-m+1}.$$
It is easy to verify that this defines a probability measure over the
set of $m$-long strings, and hence it is possible to define the
sliding block entropy in a manner analogous to the definition of the
block entropy.

It is implicit in the work of Ziv and Lempel \cite{ZivLem78} that the
sliding block entropy and the block entropy are both equal to the
finite state compressibility of a sequence. In the following section,
we establish that it is possible to attain every rational sliding
block entropy value using Liouville numbers.

\section{Finite State Dimension of Liouville Numbers}
We now have $\beta$, a Liouville number with finite-state dimension
zero, and $\alpha$, a normal Liouville number - that is, a number with
finite-state dimension 1. We show, that for any rational $q \in [0,
  1]$, we can construct a Liouville number having finite-state
dimension $q$. The construction is a variant of the standard dilution
argument in finite-state dimension \cite{Dai:FSD}.

The dilution argument is as follows. Suppose $S$ is an infinite binary
sequence, and $\frac{p}{q}$ is a rational in the unit interval
expressed in lowest terms. Then,
$$w_0 0^{q-p}\; w_1 0^{q-p}\; \dots,$$ 
where $w_0$ is the first $p$ bits of $S$, $w_1$ is the next $p$ bits
of $S$, and so on, is a binary sequence with finite-state dimension
$\frac{p}{q} \dim_{\text{FS}}(S)$. We cannot adopt this construction,
since it is not certain that a dilution of $\alpha$ gives us a
Liouville number even if it gives a sequence with the desired
finite-state dimension.

%% The Liouville construction depends crucially on the fact that the
%% $i^{\text{th}}$ stage is the repetition of a block. This enables us
%% to design a rational number sequence converging to $\alpha$ while
%% satisfying the Liouville criterion. The standard dilution may not
%% have have periodic stretches which do not suit approximation in the
%% Liouville sense.

We show that a slight variant of the dilution argument enables us to
create a Liouville number with arbitrary rational finite-state
dimension. 

Let us establish that we can construct a Liouville number with finite
state dimension $\frac{m}{n}$, where $m$ and $n$ are positive numbers,
and the rational is expressed in lowest terms.

The sequence we construct is
$$\alpha_{m/n} = 0 \cdot
                 \left((0^{2^1})^{(n-m)}\; B(2,1)^{m}\right)^{1^1} \;
%%               \left((0^{2^2})^{(n-m)}\; B(2,2)^{m}\right)^{2^2} \;
                 \dots
                 \left((0^{2^k})^{(n-m)}\; B(2,k)^{m}\right)^{k^k} \;
                 \dots$$

That is, the recurring block in the $k$th stage consists of $m$ copies
of $B(2,k)$ and a padding of $n-m$ copies of $0^{2^k}$. The recurring
block has a length of $n 2^k$. This block is then repeated $k^k$
times, to form the $k$th stage. This is similar to what happens in the
construction of the normal Liouville number.

We now show that the constructed number has the desired finite-state
dimension.

First, we count the frequency of $k$-long strings in the stage
$k$. Any string other than $0^k$ occurs $m$ times, and $0^k$ occurs
$(n-m)2^k + m$ times. 

Thus the frequency in the block, of any string other than $0^k$ is
$$\frac{m}{n 2^k},$$
and the frequency of $0^k$ is 
$$\frac{n-m}{n} + \frac{m}{n2^k}.$$

We now compute the $k$-block entropy of the $k$th stage. This is 
\begin{align}
&\frac{1}{k} \left[ \frac{m}{n2^k} (2^k - 1) \log \frac{n2^k}{m} + 
        \left(\frac{n-m}{n} + \frac{m}{n2^k}\right) \log
        \frac{n2^k}{(n-m)2^k+m} \right]\\
\begin{split}
=& \frac{1}{k} \Bigg[ \frac{m}{n2^k}2^k \log (n 2^k) -  
                      \frac{m}{n2^k} (2^k - 1) \log m  \\
 &\phantom{losungr}   -\frac{m}{n2^k} \log \left((n-m)2^k + m\right) +\;
                      \frac{n-m}{n2^k} \log
                      \frac{n2^k}{(n-m)2^k+m} \Bigg]
\end{split}
\\
\begin{split}
=& \frac{m}{n} \frac{\log (n 2^k)}{k} - 
   \frac{m}{n2^k} (2^k - 1) \frac{\log m}{k} \\
 &\phantom{losungr}-\frac{m}{n2^k} \frac{\log \left((n-m)2^k +
     m\right)}{k} +\; 
    \frac{n-m}{n2^k}\frac{1}{k} \log \frac{n2^k}{(n-m)2^k+m}\\
\end{split}
\end{align}

We now simplify the terms to get the following expression for the
$k$-block entropy of the $k^{\text{th}}$ stage.
\begin{align}
\label{eqn:bound}
\frac{m}{n} \Theta(1) - 
     \frac{m}{n} \left(1 - \Theta\left(\frac{1}{2^k}\right)\right)
     \Theta\left(\frac{1}{k}\right) - 
   \frac{m}{n2^k} \frac{\Theta(k)}{k} +
    \frac{n-m}{n2^k}\frac{1}{k} \log \frac{n2^k}{(n-m)2^k+m}.
\end{align}

The last term can be bounded using the following analysis. We know,
since $0 < m < n$, that $n2^k > (n-m)2^k \ge 2^k$. Hence, 
$$2^k < m+2^k < m+(n-m) 2^k < m+n2^k < 2n2^k, $$
the extreme terms being obtained by the bounds $0 < m < n$.

Hence,
$$\log \frac{n2^k}{2^k} > 
\log \frac{n2^k}{m+(n-m)2^k} > 
\log \frac{n 2^k}{2n2^k},$$
hence
$$\log n > \log \frac{n2^k}{m+(n-m)2^k} > \log \frac{1}{2}.$$
Since both the upper bound and lower bounds are constants independent
of $k$, we have that the last term in (\ref{eqn:bound}) is
$\Theta(1/2^k)$. 

The same bound holds for $k$-block entropy of any stage $K > k$, by
the property of $B(2,K)$. 

Thus, taking limits as $k \to \infty$, the sliding block entropy rate
of the sequence is $\frac{m}{n},$ as desired.

We now show that $\alpha_{m/n}$ is a Liouville number. We need the
following estimate for the length of $\alpha_{m/n}$ up to stage $i$.
\begin{align*}
l_i &= \sum_{m = 1}^{i} (n 2^m) m^m  
< n i^i \sum_{m=1}^{i} 2^m 
= n i^i O(2^i)
= n O([2i]^i) = O([2i]^i),
\end{align*}
noting that $n$ is a constant that does not depend on $i$.

The $i^{\text{th}}$ convergent to $\alpha_{m/n}$ is the rational 
$$\frac{\alpha_{m/n}\left[0 \dots l_{i-1}\right]}{2^{l_{i-1}}} + 
  \frac{(0^{2^i})^{(n-m)} \;B(2,i)}{(2^{n2^{i}} - 1) 2^{l_{i-1}}}.$$ 

The exponent of the denominator is $O(l_{i-1})$, but the distance of
$\alpha_{m/n}$ from the convergent is 
$$2^{-iO(l_{i-1})}.$$
This completes the proof.

%% \section{Real-valued dimensions of Liouville numbers}
%% Since we can attain a dense set of dimensions, we could inquire
%% whether it is possible for a Liouville number to possess finite-state
%% dimension equal to any real number $s \in [0,1]$.

%% A spliced construction based on the above construction suffices for
%% this purpose.

\section{A Conditional Construction Based on Artin's Conjecture}
In this section, we give an outline of how number-theoretic properties
may be employed to construct normal Liouville numbers simultaneously
normal in finitely many bases.  Let $n$ be a natural number which is
at least 2.

\begin{definition}
A number $a$ is said to be a \emph{primitive root} of a number $p$ if
the sequence $a \mod p, a^2 \mod p, a^3 \mod p, \dots, a^{p-1} \mod
p$, has $p-1$ distinct elements.
\end{definition}

For example, 2 is a primitive root of the prime 13, but is not a
primitive root of the prime 7. 

Recall that the base-$a$ expansion of any rational is eventually
periodic. We mention the following observation. If $a$ is a primitive
root of $p$, then the fraction $a/p$ when expressed in base $a$, has a
recurring block of length $p-1$ - that is, the maximal length. We
explain this observation as follows. 

To be specific, let $a = 2$ and $m$ be an arbitrary positive
integer. Let $1/m$ consist of repetitions of the binary string $b_1
\dots b_{k}$. Then the length of the recurring block ($k$), is equal
to the number of times the binary expansion has to be shifted left
before the fractional part represents $1/m$ again. The orbit of the
number $1/m$ under the left-shift is precisely the set
%===========
$$\left\{ \frac{2}{m} \mod 1, \dots, \frac{2^{m-1}}{m} \mod 1\right\} =
\left\{ \frac{2 \mod m}{m} , \dots, \frac{2^{m-1} \mod m}{m}\right\}.$$
%===========
By the discussion above, the maximal $k$, equal to $m-1$, is attained
when $2$ is a primitive root of $m$. The following lemma says that if
$a$ is a primitive root of $p$, the left-shifts of the base $a$
expansion of $1/p$ show certain degree of uniformity in
distribution. Thus the base-$a$ expansion of $1/p$ is the analogue of
an $a$-ary deBruijn sequence.

\begin{lemma}
Let $a$ be a primitive root of a prime number $p$. Then for each $k$
with $a^k < p$, for each $0 \le j < a^k$, and each positive integer
$n$, we have the following
frequency estimate:
\begin{align}
\label{eqn:u_bound}
\frac{\left|\left\{ \frac{a^m}{p} \mod 1 \in \left[\frac{j}{a^k},
    \frac{j+1}{a^k}\right) \mid n\le m \le n+p-1 \right\} \right|}{p-1}
  = \frac{1}{a^k} + \Theta\left(\frac{1}{(p-1)a^k}\right).
\end{align}
\end{lemma}
\begin{proof}
The number of orbit points $\frac{a^{m}}{p} \mod 1$, $n \le m \le n+p-1$
which fall into $\left[ \frac{j}{a^k}, \frac{j+1}{a^k} \right)$ are
  precisely those for which the following inequality hold.

$$\frac{j}{a^k} \le \frac{a^{m}}{p} \mod 1 < \frac{j+1}{a^k}.$$ 

Since $p$ is a prime, we have $a^k \ne p$. Thus we have either
$\lfloor \frac{p-1}{a^k} \rfloor$ or $\lceil \frac{p-1}{a^k} \rceil$
numbers in $1 \le m \le p-1$ which satisfy the above inequality. From
this, the required density estimate follows.
\begin{flushright}\qed\end{flushright}
\end{proof}

The above lemma has the consequence that for each fixed $k$, as $p \to
\infty$, the frequency in (\ref{eqn:u_bound}) tends to $1/a^k$.

We could construct a normal number along the lines of section 5 if
we are assured of an infinite number of such $p$ for each $a$. This is
a consequence of \emph{Artin's Conjecture:}
\begin{conjecture}\cite{Artin65}
Let $a$ be an integer other than -1, and a perfect square. Then for
any $x \in \N$, the number of primes for which $b$ is a primitive root
is asymptotically $A(a) \frac{x}{\log x},$ where $A(a)$ is a constant
dependent on $a$.
\end{conjecture}

Heath-Brown \cite{HB85} proved that the conjecture is true for all
prime $a$ with at most two exceptions, and for any $x \in \N$, the
number of $a \in \Z$, $|a| \le x$, for which the conjecture fails is
$o((\log x)^2)$. We use a generalized version of this conjecture by
Matthews.

\begin{conjecture}\cite{Matthews76}
Let $a_1, a_2, \dots, a_n$ be non-zero integers not $\pm 1$. Then for
any $x \in \N$, the number of primes $\le x$ is asymptotically
$$\frac{x}{\log x} A(a_1, \dots, a_n) + O\left(\frac{x}{\log^2
  x} (\log\log x)^{2^n - 1}\right).$$
\end{conjecture}

Matthews \cite{Matthews76} established this conjecture assuming a
special case of the Generalized Riemann Hypothesis. We call this the
\emph{Generalized Artin's Conjecture}. We utilize this to outline a
construction of an Liouville number simultaneously normal to finitely
many distinct bases $1 < a_1 < a_2 < \dots < a_n$. The construction
relies on the following property of base-$k$ expansions.

\begin{lemma}
Let $a_1, \dots, a_n$ be primitive roots of $p$. Then $1/p \times
(a_1a_2 \dots a_n)^{p-1} \mod 1 = 1/p$.
\end{lemma}
\begin{proof}
We prove the assertion in the case when $a_1, a_2$ are distinct
primitive roots of $p$. We know that 
%============
$$\frac{1}{p} a_1^{p-1} \mod 1 = \frac{1}{p} a_1^{p-1} \mod 1 =
\frac{1}{p},$$
%============
since $a_1$ and $a_2$ are primitive roots of $p$. Then, 
%============
$$\frac{1}{p} a_2^{p-1} \mod 1 =
\left(\frac{1}{p} a_1^{p-1} \mod 1 \right) a_2^{p-1} \mod 1 = 
\frac{1}{p} (a_1 a_2)^{p-1} \mod 1.$$
%============
The general case proceeds by induction on $n$.
\begin{flushright}\qed\end{flushright}
\end{proof}

\begin{corollary}
\label{cor:ud}
Let $1 < a_1 < \dots < a_n$ be numbers which are primitive roots of
$p$. Then for any $a_i$, $1 \le i \le n$, we have 
$$\left\{ \frac{a_i^k}{p} \mod 1 \mid 0 \le k \le p-1\right\} =
\left\{ \frac{a_i^k\times(a_1 a_2 \dots a_n)^{p-1}}{p} \mod 1 \mid 0
\le k \le p-1\right\}.$$
\end{corollary}
\begin{proof}
This is true because for any $a_i$ and any positive number $k$,
$$a_i^k \mod p = [a_i^k (a_1 a_2 \dots a_n)^{p-1}] \mod p.$$
\begin{flushright}\qed\end{flushright}
\end{proof}

We now ``left-shift'' $1/p$ enough to extract some repetition of the
recurrent block in the expansions in all the bases $a_1, \dots, a_n$.
For any positive integer $i$, let $p_i$ be the $i^\text{th}$ largest
prime such that $a_1, \dots, a_n$ are simultaneously the primitive
roots of $p_i$, assuming such primes exist. We denote the integer
%============
$$P_i = \left\lfloor \frac{1}{p_i} (a_1 a_2 \dots a_n)^{p_i
  - 1}\right\rfloor.$$
%============
For the following construction, the basic building blocks will be the
``right-shifted'' versions of $P_i$.  Denote
%============
$$N(i) = P_i \times (a_1 a_2 \dots a_n)^{-(p_i - 1)}.$$ 
%============
By construction, for any $1 \le i \le n$, the representation of $N(i)$
in base $a_i$, will be a repetitive block of digits, followed by an
infinite sequence made of zeroes.

Let $f: \N \to \N$ be a function which we will specify later. Now we
construct the stage $i$ as $f(i)$ repetitions of the non-zero digits
in the expansion of $N(i)$. Consider
\begin{multline*}
S'(i) = N(i) \times (a_1 a_2 \dots a_n)^{-(p_i-1)} + 
          N(i) \times (a_1 a_2 \dots a_n)^{-2(p_i-1)} +  
          \dots +\\
          N(i) \times (a_1 a_2 \dots a_n)^{-f(i) \times (p_i-1)},
\end{multline*}
and the ``appropriately right-shifted'' version of $S'(i)$,
$$S(i) = S'(i) \times (a_1 a_2 \dots a_n)^{ \sum_{j=1}^{i-1}-f(j)(p_j-1)}.$$

Then the Liouville number normal in bases $1 < a_1 < \dots < a_n$ is 
$$\gamma = \sum_{i=1}^{\infty}  S(i).$$

We need only the consequence that for the bases $a_1, \dots, a_n$,
there are infinitely many such primes, for the above
construction. (The density estimates help lower bound the stage
lengths, but are not strictly needed for the construction.) Let $a_1,
a_2, \dots, a_n$ be numbers for which the Generalized Artin's
conjecture holds - thus there are infinitely many primes $p_1, p_2,
\dots$ such that each $a_i$ is a primitive root of each of these
primes.

To ensure that $\gamma$ is a Liouville number, we need $f$ to satisfy
the following condition: for each $i \ge 2$, the cumulative length of
the stages $1,\dots, i-1$ should be greater than the number of
non-zero digits in the expansion of $N(i+1)$ in any of the bases $a_1,
\dots, a_n$. The length of $N(i)$ is at most $\log_{a_1} p_{i+1}$. The
length of the expansion of the block in the $i^\text{th}$ stage is at
least $\log_{a_n} p(i)$. So it is sufficient to ensure that $f$
satisfies, for each $i \ge 2$, 
$$ \log_{a_1} (p_{i+1}) = o(\log_{a_n} (p_i) \times f(i)).$$ 
This property can be ensured by making $f(i)$ large enough, assuming
the Generalized Artin's conjecture.

The proof that $\gamma$ is a Liouville number follows from the
property of $f$. That $\gamma$ is a normal number in each of the bases
$2 < a_1 < \dots < a_n$ will follow from Corollary \ref{cor:ud}.

\begin{acknowledgements}
The first author would like to gratefully acknowledge the help of
David Kandathil and Sujith Vijay, for their suggestions during early
versions of this work and Mrinal Ghosh, Aurko Roy and anonymous
reviewers for helpful suggestions.
\end{acknowledgements}

\bibliographystyle{plain}
\bibliography{dim,dimrelated,fair001,main,random}

\end{document}